\documentclass[11pt]{article}
\usepackage{amsmath,amssymb,amsbsy,amsgen,amsopn,amsthm}
\usepackage{mathrsfs}
\usepackage{hyperref}
\usepackage[margin=1in]{geometry}

\newtheorem{theorem}{Theorem}
\newtheorem{lemma}{Lemma}
\newtheorem{fact}{Fact}

\newcommand{\Rset}{\mathbb{R}}
\newcommand{\Efam}{\mathcal{E}}
\newcommand{\LP}{P}
\newcommand{\Dual}{D}
\newcommand{\EC}{EC}

\newcommand{\TreeDual}{D}

\newcommand{\lca}{{\rm lca}}

\newcommand{\Var}{{\tt Var}}

\usepackage{tikz}
\usetikzlibrary{backgrounds}
\usetikzlibrary{snakes}
\usetikzlibrary{shapes}
\usetikzlibrary{trees}
\tikzstyle{vertex}=[circle,draw, fill=white,inner sep=2.3pt]
\tikzstyle{terminal}=[circle,draw,fill, inner sep=2.3pt]
\tikzstyle{biset}=[fill=gray,fill opacity=.2, line width=8pt, draw=gray]
\tikzstyle{biset outer}=[dotted,line width=2pt]

\title{Covering problems in edge- and node-weighted graphs}

\author{Takuro Fukunaga\footnote{National Institute of Informatics,
2-1-2 Hitotsubashi, Chiyoda-ku, Tokyo, Japan.
JST, ERATO, Kawarabayashi Large Graph Project, Japan.
Email: takuro@nii.ac.jp}}

\date{}

\begin{document}
\maketitle

\begin{abstract}
 This paper discusses the graph covering problem in which a set of edges
 in an edge- and node-weighted graph 
is chosen to satisfy some covering constraints while minimizing the sum
 of the weights.
In this problem, because of the large integrality gap of a
natural linear programming (LP) relaxation, 
 LP rounding algorithms based on the
relaxation yield poor performance.
 Here we propose a stronger LP relaxation for the graph covering problem.
 The proposed relaxation is applied to designing primal-dual algorithms 
 for two fundamental graph covering problems: the prize-collecting edge
 dominating set problem and the multicut problem in trees.
 Our algorithms
 are an exact polynomial-time
 algorithm for the former problem, and a 2-approximation algorithm for
 the latter problem, respectively.
 These results match the currently known best
 results for purely edge-weighted graphs.
\end{abstract}

\section{Introduction}
\label{sec.intro}

\subsection{Motivation}
Choosing a set of edges in a graph that optimizes some objective function
under constraints on the chosen edges
constitutes a typical combinatorial optimization problem and has been investigated in many varieties.
For example, the spanning tree problem seeks an acyclic
edge set that spans all nodes in a graph, the edge cover problem
finds an edge set such that each node is incident to at
least one edge in the set, and the shortest path problem selects an edge set that connects two specified nodes.
All these problems seek to minimize the sum of the weights assigned to
edges.

This paper discusses several graph covering problems.
Formally, the \emph{graph covering problem} is defined as follows in
this paper.
Given a graph $G=(V,E)$
and family $\Efam \subseteq 2^E$,
find a subset $F$ of $E$ that satisfies
$F \cap C \neq \emptyset$ for each $C \in \Efam$, while optimizing
some function depending on $F$.
As indicated above, 
the popular approaches assume an edge weight function
$w \colon E \rightarrow \Rset_+$ is given, where $\Rset_+$ denotes the set of
non-negative real numbers, and
seeks to minimize $\sum_{e \in F}w(e)$.
On the other hand, we aspire to simultaneously minimize edge and node weights.
Formally, we let $V(F)$ denote the set of end nodes of edges in $F$.
Given a graph $G=(V,E)$ and weight function $w\colon E\cup V \rightarrow \Rset_+$,
we seek a subset $F$ of $E$ that 
minimizes $\sum_{e\in F}w(e) + \sum_{v \in V(F)}w(v)$ under the
constraints on $F$.
Hereafter, we denote $\sum_{e\in F}w(e)$ and $\sum_{v \in
V(F)}w(v)$ by $w(F)$ and $w(V(F))$, respectively.

Most previous investigations of the graph covering problem have focused
on edge weights. By contrast, node weights have been largely neglected, except in the
problems of choosing node sets, such as the vertex cover and dominating set problems.
To our knowledge, 
when node weights have been considered in graph covering problems
for choosing edge sets, they have been
restricted to the Steiner tree problem or its generalizations,
possibly because the inclusion of node weights greatly complicates the
problem. For example, the Steiner tree problem in edge-weighted
graphs can be
approximated within a constant factor (the best currently known
approximation factor is 1.39~\cite{Byrka2013,Goemans2012}). Conversely, 
the Steiner tree problem in node-weighted graphs
is known to extend the set cover problem (see \cite{KleinR95}), indicating that achieving an
approximation factor of $o(\log |V|)$ is NP-hard.
The literature is reviewed in Section~\ref{sec.related}. As revealed later, 
the inclusion of node weights generalizes the set cover problem in numerous fundamental problems.

However, from another perspective, node weights can introduce rich structure
into the above problems.
In fact, node weights provide useful optimization problems.
The objective function counts the weight of a node only once, even if the node is shared by multiple edges.
Hence, the objective function defined from node weights includes a certain
subadditivity, which cannot be captured by edge weights.

The aim of the present paper is to give algorithms for fundamental graph covering problems in
edge- and node-weighted graphs. In solving the problems, we adopt a basic linear programming (LP)
technique. Many algorithms for combinatorial
optimization problems are typically designed using LP
relaxations. However, in problems with node-weighted graphs, the integrality gap of natural relaxations may be excessively large.
Therefore, we propose tighter LP relaxations that preclude unnecessary integrality gaps.
We then discuss upper bounds on the integrality gap of these relaxations in 
two fundamental graph covering problems:
the edge dominating set (EDS) problem and multicut problem in trees. 
We prove upper bounds by designing primal-dual algorithms for both problems. The approximation factors of our proposed algorithms match the current best approximations in purely edge-weighted graphs.

\subsection{Problem definitions}

The EDS problem covers edges by
choosing adjacent edges in undirected graphs.
For any edge $e$, let $\delta(e)$ denote the set of edges that share
end nodes with $e$, including $e$ itself.
We say that an edge $e$ {\em dominates} another edge $f$ if $f \in
\delta(e)$,
and a set $F$ of edges dominates an edge $f$ if $F$
contains an edge that dominates $f$.
Given an undirected graph $G=(V,E)$, 
a set of edges is called an EDS
if it dominates each edge in $E$.
The EDS problem seeks to minimize the weight of the EDS.
In other words, the EDS problem is the graph covering
problem with $\Efam=\{\delta(e) \colon e \in E\}$.

In the multicut problem, an instance specifies an undirected graph $G=(V,E)$ and demand pairs
$(s_1,t_1),\ldots,(s_k,t_k) \in V\times V$.
A {\em multicut} is an edge set $C$ whose removal from $G$ disconnects
the nodes in each demand pair. This problem seeks a multicut of minimum weight. 
Let $\mathcal{P}_i$ denote the set of paths connecting $s_i$ and $t_i$.
The multicut problem is equivalent to the graph covering problem with
$\Efam=\bigcup_{i=1}^k \mathcal{P}_i$.

Our proposed algorithms for solving these problems assume
that the given graph $G$ is a tree.
In fact, our algorithms are applicable to the 
prize-collecting versions of these problems, which additionally specifies
a penalty function $\pi \colon \Efam \rightarrow \Rset_+$. 
In this scenario, an edge set $F$ is a feasible solution even if $F\cap C
= \emptyset$ for some $C \in \Efam$, but imposes a penalty
$\pi(C)$.
The objective is to minimize the sum of $w(F)$, $w(V(F))$, and the
penalty $\sum_{C \in \Efam: F \cap C=\emptyset}\pi(C)$.
The prize-collecting versions of the EDS and multicut problems 
are referred to as the {\em prize-collecting EDS problem}
and the {\em prize-collecting multicut problem}, respectively.

\subsection{Our results}
Thus far, the EDS problem has been applied only to
edge-weighted graphs.
The vertex cover problem can be reduced to
the EDS problem while preserving the approximation factors~\cite{Carr2001}.
The vertex cover problem is solvable by a 2-approximation algorithm, which is
widely regarded as the best possible approximation. Indeed, assuming the unique game conjecture, Khot and Regev~\cite{KhotR08}
proved that the vertex cover problem cannot be approximated within a factor better than $2$.
Fujito and Nagamochi~\cite{Fujito2002} showed that 
a 2-approximation algorithm is admitted by the EDS problem, which matches
the approximation hardness known for the vertex cover problem.
In the Appendix, we show that the EDS
problem in bipartite graphs
generalizes the set cover problem if assigned node weights and
generalizes the non-metric facility location problem if assigned edge
and node weights. 
This implies that including node weights increases difficulty of the
problem even in bipartite graphs.

On the other hand, Kamiyama~\cite{Kamiyama2010} proved that the
prize-collecting EDS problem in an edge-weighted graph
admits an exact polynomial-time algorithm if the graph is a tree.
As one of our main results, we show that this idea is extendible to
problems in edge- and node-weighted trees.

\begin{theorem}\label{thm.eds}
 The prize-collecting EDS problem 
 admits a polynomial-time exact algorithm for edge- and node-weighted trees.
\end{theorem}

Theorem~\ref{thm.eds} will be proven in Section~\ref{sec.eds}.
As demonstrated in the Appendix, the prize-collecting EDS problem
in general edge- and node-weighted graphs admits an $O(\log |V|)$-approximation, which matches
the approximation hardness
on the set cover problem and the non-metric facility location problem.

The multicut problem is hard even in edge-weighted graphs; the best
reported approximation factor is $O(\log k)$~\cite{Garg1996}. The
multicut problem is known to be both
NP-hard and MAX SNP-hard~\cite{Dahlhaus1994}, 
and admits no constant factor approximation algorithm under the unique game conjecture~\cite{Chawla2006}.
However,
Garg, Vazirani, and Yannakakis~\cite{Garg1997} 
developed a 2-approximation algorithm
for the multicut problem with edge-weighted trees.
They also mentioned that, although the graphs are restricted to trees, 
the structure of the problem is sufficiently rich.
They showed that the tree multicut problem includes the set cover problem with tree-representable set systems.
They also showed that the vertex cover problem in general graphs is simply reducible to the multicut problem in star graphs, while preserving the approximation factor. This
implies that the 2-approximation seems to be tight for the multicut problem in trees.
As a second main result, we extended this 2-approximation to 
edge- and node-weighted trees, as stated in the following theorem.

\begin{theorem}\label{thm.multicut}
 The prize-collecting multicut problem
 admits a 2-approximation algorithm for edge- and node-weighted trees.
\end{theorem}

Both algorithms claimed in Theorems~\ref{thm.eds} and \ref{thm.multicut} are primal-dual algorithms,
that use the LP relaxations we propose. 
These algorithms fall in the same frameworks as
those proposed in~\cite{Garg1997,Kamiyama2010} for edge-weighted graphs.
However, they need several new ideas to achieve the claimed performance because our LP relaxations
are much more complicated than those used in~\cite{Garg1997,Kamiyama2010}. 

The remainder of this paper is organized as follows.
After surveying related work in Section~\ref{sec.related},
we define our LP relaxation for the prize-collecting graph covering problem in Section~\ref{sec.lp}.
Using this relaxation, we prove Theorems~\ref{thm.eds} and~\ref{thm.multicut} in Sections~\ref{sec.eds} and~\ref{sec.multicut}, respectively.
The paper concludes with Section~\ref{sec.conclusion}.
In the Appendix, we show that the prize-collecting EDS
problem in edge- and node-weighted graphs generalizes
the set cover problem and facility location problem,
and admits an $O(\log |V|)$-approximation algorithm.

\section{Related work}\label{sec.related}

As mentioned in Section~\ref{sec.intro},
the graph covering problem in node-weighted graphs has thus far been applied to the
Steiner tree problem and its generalizations.
Klein and Ravi~\cite{KleinR95} proposed an $O(\log |V|)$-approximation
algorithm for the Steiner tree problem with node weights.
Nutov~\cite{Nutov10node-weights,Nutov12uncrossable} extended this algorithm to 
the survivable network design problem with
higher connectivity requirements.
An $O(\log |V|)$-approximation algorithm for the prize-collecting Steiner
tree problem with node weights was provided by Moss and Rabani~\cite{MossR07}; however, as
noted by K\"onemann, Sadeghian, and Sanit\`a~\cite{Konemann13CoRR}, the proof of this algorithm contains a technical error. This error was corrected in~\cite{Konemann13CoRR}.
Bateni, Hajiaghayi, and Liaghat~\cite{BateniHL13} proposed an $O(\log
|V|)$-approximation algorithm for the prize-collecting Steiner forest
problem and applied it to the budgeted Steiner tree problem.
Chekuri, Ene, and Vakilian~\cite{ChekuriEV12} gave an 
$O(k^2 \log |V|)$-approximation algorithm for
the prize-collecting survivable
network design problem with edge-connectivity requirements of maximum
value $k$. Later, they improved their approximation factor to $O(k \log |V|)$,
and also extended it to
node-connectivity requirements (see \cite{Vakilian13}).
Naor, Panigrahi, and Singh~\cite{DBLP:conf/focs/NaorPS11} established
an online algorithm for the Steiner tree problem with node weights
which was extended to the Steiner forest problem by Hajiaghayi, Liaghat, and Panigrahi~\cite{Hajiaghayi2013}.
The survivable network design problem with
node weights has also been extended to a problem called the network activation
problem 
\cite{Panigrahi11wireless,Nutov12activation,Fukunaga2014}.

The prize-collecting EDS problem generalizes the
$\{0,1\}$-EDS problem, in which given demand edges require being
dominated by a solution edge set.
The $\{0,1\}$-EDS problem in general edge-weighted graphs admits a
$8/3$-approximation, which was proven by
Berger~et~al.~\cite{Berger2007}.
This $8/3$-approximation
was extended to the prize-collecting EDS problem by Parekh~\cite{Parekh2008}.
Berger and Parekh~\cite{Berger2008} designed an exact algorithm for the $\{0,1\}$-EDS problem in
edge-weighted trees, but 
their result contains an error~\cite{Berger2012}.
Since the prize-collecting EDS problem embodies the $\{0,1\}$-EDS problem,
the latter problem could be alternatively solved by an algorithm
developed for the prize-collecting EDS problem in edge-weighted trees, proposed by
Kamiyama~\cite{Kamiyama2010}.

\section{LP relaxations}
\label{sec.lp}

This section discusses LP relaxations for the prize-collecting 
graph covering problem in edge and node-weighted graphs.

In a natural integer programming (IP) formulation of the graph covering problem,
each edge $e$ is associated with a variable $x(e)\in \{0,1\}$, and each node $v$ is associated with a
variable $x(v) \in \{0,1\}$.
$x(e)=1$ denotes that $e$ is selected as part of the solution set, while
$x(v)=1$ indicates the selection of an edge incident to $v$.
In the prize-collecting version, 
each demand set $C \in \Efam$ is also associated with a variable $z(C) \in \{0,1\}$,
where $z(C)=1$ indicates that the covering constraint corresponding
to $C$ is not satisfied.
For $F\subseteq E$,
we let $\delta_F(v)$ denote the set of edges incident to $v$ in $F$. The subscript may be
removed when $F=E$.
An IP of the prize-collecting graph
covering problem is then formulated as follows.
\begin{align*}
&\mbox{minimize} && 
 \sum_{e \in E} w(e)x(e) + \sum_{v \in V}w(v)x(v) + \sum_{C \in \Efam}\pi(C)z(C)\\
 &\mbox{subject to} &&
 \sum_{e \in C} x(e) \geq 1-z(C) && \mbox{for } C \in \Efam,\\
 &&& x(v) \geq x(e) && \mbox{for } v \in V, e \in \delta(v),\\
 &&& x(e) \geq 0 && \mbox{for } e \in E,\\
 &&& x(v) \geq 0 && \mbox{for } v \in V,\\ 
 &&& z(C) \geq 0 && \mbox{for } C \in \Efam.
\end{align*}
In the above formulation, the first constraints specify the covering
constraints, while the second constraints indicate that if the solution contains an edge $e$
incident to $v$, then $x(v)$ = 1.
In the graph covering problem (without penalties),
$z$ is fixed at 0.

To obtain an LP relaxation, we relax the definitions of $x$
and $z$ in the above IP 
to $x \in \Rset_+^{E \cup V}$ and $z \in \Rset_+^C$.
However, this relaxation may introduce a large integrality gap 
into the graph covering problem with node-weighted graphs, as shown in the following example.
Suppose that $\Efam$ comprises a single edge set $C$,
and each edge in $C$ is incident to a node $v$.
Let the weights of all edges and nodes other than $v$ be 0.
In this scenario, the optimal value of the graph covering problem is $w(v)$.
On the other hand, the LP relaxation admits a feasible solution $x$ such
that $x(v)=1/|C|$ and
$x(e)=1/|C|$ for each edge $e \in C$. 
The weight of this solution is $w(v)/|C|$,
and the integrality gap of the relaxation for this instance
is $|C|$.

This phenomenon occurs even in the EDS problem
and multicut problem in trees.
For instance, consider a star of $n$ leaves in the EDS problem.
The weight of all edges and nodes is 0 except the center node
$v$, whose weight is 1.
In this instance of the EDS problem, the weight of any
EDS is 1. On the other hand, LP relaxation admits a feasible
solution $x$ such that
$x(e)=1/n$ for each edge $e$, and $x(u)=1/n$ for each node $u$.
Since the weight of this fractional solution is $1/n$, the integrality
gap is $n$.

Let $uv$ denote an edge that joints nodes $u$ and $v$.
To determine the integrality gap in the multicut problem, 
we consider that each edge $uv$ in the star is
subdivided into two edges $us$ and $sv$.
The subdivision imposes a weight of 1 on node $s$.
All edges and remaining nodes (i.e., the center node and all leaves) 
have weight 0.
All pairs of leaves are demand pairs.
A path between the center node and a leaf is called a \emph{leg}.
In this instance, any multicut must choose at least one edge from each
of $n-1$ legs. Hence, the minimum multicut weight is $n-1$.
On the other hand, if $x(e)=1/4$ for every edge $e$ and $x(v')=1/4$ for
every node $v'$, the weight
is $n/4$ (such a fractional solution is feasible to the relaxation). Hence, the integrality gap of the relaxation is at least $4$.
By contrast, Garg, Vazirani, Yannakakis~\cite{Garg1997} proved that
the integrality gap of the relaxation is at most 2 when node weights are not considered.

The above poor examples can be excluded if the second
constraints in the relaxation are replaced by $x(v) \geq \sum_{e \in
\delta(v)}x(e)$ for $v \in V$.
However, the LP obtained by this modification does not relax the graph covering
problem
if the optimal solutions contain high-degree nodes.
Thus, we introduce a new variable $y(C,e)$ for each pair
of $C \in \Efam$ and $e \in C$,
and replace the second constraints by $x(v)\geq \sum_{e \in
\delta(v)}y(C,e)$, where
$v \in V$ and $C\in \Efam$.
$y(C,e)=1$ indicates that $e$ is chosen to satisfy the covering
constraint of $C$, and $y(C,e)=0$ implies the opposite.
Roughly speaking, $y(C,\cdot)$ represents a minimal fractional solution
for covering a single demand set $C$. If a single covering constraint is imposed, the degree of each node is
at most one in any minimal integral solution. Then
the graph covering problem is relaxed by the LP even after
modification. 
Summing up, we formulate our LP relaxation for 
an instance $I=(G,\Efam,w,\pi)$ of
the prize-collecting graph covering problem as follows.
\begin{align*}
&\LP(I)=\\
&\mbox{minimize} && 
 \sum_{e \in E} w(e)x(e) + \sum_{v \in V}w(v)x(v) + \sum_{C \in \Efam}\pi(C)z(C)\\
 &\mbox{subject to} &&
 \sum_{e \in C} y(C,e) \geq 1 -z(C)&& \mbox{for } C \in \Efam,\\
 &&& x(v) \geq \sum_{e\in \delta_C(v)}y(C,e) && \mbox{for } v \in V, C \in \Efam,\\
 &&& x(e) \geq y(C,e)  && \mbox{for } C \in \Efam, e \in C, \\
 &&& x(e) \geq 0&&  \mbox{for } e \in E, \\
 &&& x(v) \geq 0 &&  \mbox{for } v \in V,\\
 &&& y(C,e) \geq 0 && \mbox{for } C \in \Efam, e \in C,\\
 &&& z(C) \geq 0 && \mbox{for } C \in \Efam.
\end{align*}

\begin{theorem}
 Let $I$ be an instance of the 
 prize-collecting graph covering problem in edge- and node-weighted graphs.
 $\LP(I)$ is not greater than the optimal value of $I$.
\end{theorem}
\begin{proof}
 Let $F$ be an optimal solution of $I$.
 We define a solution $(x,y,z)$ of $\LP(I)$ from $F$.
 For each $C \in \Efam$, we set $z(C)$ to 0 if $F \cap C\neq
 \emptyset$,
 and $1$ otherwise.
 If $F\cap C \neq \emptyset$,
 we choose an arbitrary edge $e \in F\cap C$, and let
 $y(C,e)=1$.
 For the remaining edges $e'$, we assign $y(C,e')=0$.
In this way, the values of variables in $y$ are defined for each $C \in \Efam$.
 $x(e)$ is set to 1 if $e \in F$, and 0 otherwise.
 $x(v)$ is set to 1 if $F$ contains an edge incident to $v$, and 
0 otherwise.

 For each $C \in \Efam$ with $z(C)<1$, exactly one edge $e$ satisfies $y(C,e)=1$, and this $e$ is included in $F \cap C$.
 If $y(C,e)=1$, then $x(e)=1$, and
 each end node $v$ of $e$ satisfies $x(v)=1$.
 For a pair of $v \in V$ and $C \in \Efam$,
 $y(C,e)$ is one for exactly one edge $e \in \delta_C(v)$,
 and zero for the remaining edges in $\delta_C(v)$.
 Thus, $(x,y,z)$ is feasible.
 The objective value of $(x,y,z)$ in $\LP(I)$ is given by
 $w(F)+w(V(F))+\sum_{C \in \Efam: F\cap C=\emptyset}\pi(C)$, which is
 the optimal value of $I$, and the theorem is proven.
\end{proof}

In some graph covering problems, $\Efam$ is not explicitly given,
and $|\Efam|$ is not bounded by a polynomial on the input size of the problem.
In such cases, the above LP may not be solved in polynomial time because
it cannot be written compactly.
However, in this scenario, we may define a tighter LP than the natural relaxation
if we can find 
$\Efam_1,\ldots,\Efam_t \subseteq \Efam$ such that 
$\cup_{i=1}^t \Efam_i = \Efam$, 
$t$ is bounded by a polynomial of input size,
and the degree of each node is small in
any minimal edge set covering all demand sets in $\Efam_i$ for each $i
\in \{1,\ldots,t\}$.
Applying these conditions, the present author obtained a new
approximation algorithm for solving
a problem generalizing some prize-collecting graph covering problems~\cite{Fukunaga2014}.

\section{Prize-collecting EDS problem in trees}\label{sec.eds}
In this section, we prove Theorem~\ref{thm.eds}.
We regard the input graph $G$ as a rooted tree,
with an arbitrary node $r$ selected as the root.
The \emph{depth} of a node $v$ is the number of edges on the path between $r$ and $v$.
When $v$ lies on the path between $r$ and another node $u$,
we say that $v$ is an \emph{ancestor} of $u$ and $u$ is a \emph{descendant} of $v$.
If the depth of node $v$ is the maximum among all ancestors of $u$, then $v$ is defined as the \emph{parent} of $u$.
If $v$ is the parent of $u$, then $u$ is a \emph{child} of
$v$.
The upper and lower end nodes of an edge $e$ are denoted by $u_e$
and $l_e$, respectively.
We say that an edge $e$ is an ancestor 
of a node $v$ and $v$ is a descendant
of $e$ when $l_e=v$ or $l_e$ is an ancestor of $v$. Similarly, an edge $e$ is a descendant
of a node $v$ and $v$ is an ancestor
of $e$ if $v=u_e$ or $v$ is an ancestor of $u_e$.
An edge $e$ is defined as an ancestor of another edge $f$
if $e$ is an ancestor of $u_f$.

Recall that $\Efam=\{\delta(e)\colon e \in E\}$ in the EDS problem.
Let $I=(G,w,\pi)$ be an instance of the prize-collecting EDS problem.
We denote $\bigcup_{e \in \delta(v)}\delta(e)$ by $\delta'(v)$
for each $v\in V$.
Then the dual of $\LP(I)$ is formulated as follows.
  \begin{align}
  &{\Dual}(I)=\hspace*{-2em} \notag \\
  &\mbox{maximize} && 
 \sum_{e \in E} \xi(e) \notag\\
  &\mbox{subject to} &&
   \sum_{e \in \delta(e')}\nu(e',e) \leq w(e') && \mbox{for } e' \in E, \label{eds.d1}\\
   &&&\sum_{e \in \delta'(v)}\mu(v,e) \leq w(v) && \mbox{for } v \in V,\label{eds.d2}\\
   &&&\xi(e) \leq \mu(u,e) + \mu(v,e) + \nu(e',e) && \mbox{for } e\in E, e'=uv  \in \delta(e), \label{eds.d3}\\
   &&& \xi(e) \leq \pi(e) && \mbox{for } e \in E,\label{eds.d4}\\
   &&&\xi(e)\geq 0  && \mbox{for } e \in E,\notag\\
   &&&\nu(e',e) \geq 0 && \mbox{for } e' \in E, e \in \delta(e'),\notag\\
   &&&\mu(v,e) \geq 0&& \mbox{for } v \in V, e \in \delta'(v).\notag
 \end{align}

For an edge set $F \subseteq E$, let $\tilde{F}$ denote $\{e \in E \colon
\delta_F(e)=\emptyset\}$, and let $\pi(\tilde{F})$ denote $\sum_{e \in \tilde{F}}\pi(e)$.
For the instance $I$, our algorithm yields
a solution $F\subseteq E$ 
and a feasible solution $(\xi,\nu,\mu)$ to $\Dual(I)$,
both satisfying
\begin{equation}\label{eq.eds-optimiality}
w(F) + w(V(F)) + \pi(\tilde{F})
 \leq \sum_{e \in E}\xi(e).
\end{equation}
Since the right-hand side of \eqref{eq.eds-optimiality}
is at most $\LP(I)$,  $F$ is an optimal solution of
$I$.
We note that the dual solution $(\xi,\nu,\mu)$ is required only for
proving the optimality of the solution and need not be computed.

The algorithm operates by induction on the number of nodes of
depth exceeding one.
In the base case, all nodes are of depth one, indicating that $G$ is a star centered at $r$.
The alternative case is divided into two sub-cases:
Case A, in which a leaf
edge $e$ of maximum depth satisfies $\pi(e)>0$;
and Case B, which contains no such leaf edge.

\subsubsection*{Base case}
In the base case, $G$ is a star
centered at $r$.
Note that all edges in this graph are adjacent.
Let $\alpha_1=\min_{rv \in E}\{w(rv)+w(r)+w(v)\}$ and $\alpha_2=\sum_{e \in
E}\pi(e)$.
An edge $rv$ attaining $\alpha_1=w(rv)+w(r)+w(v)$ is denoted by $e^*=rv^*$.

If $\alpha_1 \geq \alpha_2$,
our algorithm sets $F$ as $\emptyset$,
and defines $\xi(e)$ as $\pi(e)$ for each $e \in E$.
Otherwise (i.e., $\alpha_1 < \alpha_2$), 
it specifies $F$ as $\{e^*\}$, and 
sets $\xi$ so that 
$\sum_{e\in E}\xi(e)=\alpha_1$, and $\xi(e)\leq \pi(e)$ for each $e \in E$,
which is possible because $\sum_{e \in E}\pi(e) =\alpha_2 > \alpha_1$.
Note that $F$ and $\xi$ defined in this way satisfy \eqref{eq.eds-optimiality}.

To completely define the dual solution,
we must define variables $\nu$ and 
$\mu$. Let $rv \in E$. To satisfy \eqref{eds.d3} for $e \in E$
and $rv$, 
the sum of $\mu(r,e)$, $\mu(v,e)$, and $\nu(rv,e)$
cannot be smaller than $\xi(e)$ for each $e\in E$.
Note that in \eqref{eds.d1}, $\nu(rv,e)$ is bounded from above for
$rv$, 
while in \eqref{eds.d2} for $v$, $\mu(r,e)$ and $\mu(v,e)$ are bounded for $r$
and $v$, respectively.
As an alternative interpretation, each $rv \in E$ has capacity $w(rv)$ shared by $\nu(rv,e)$,
$e \in \delta(rv)$, and each $v \in V$ has 
capacity $w(v)$ shared by $\nu(v,e)$, $e \in \delta'(v)$.
The following lemma claims that $\nu$ and $\mu$ may be set to satisfy all of
these constraints. 

\begin{lemma}\label{lem.star}
 Suppose that $G$ is a star.
 If $\alpha_1\geq \alpha_2$,
 let $\xi(e)=\pi(e)$ for each $e \in E$.
 Otherwise,
 suppose that $\xi$ is defined to satisfy 
 $\sum_{e \in E}\xi(e)=\alpha_1$, and $\xi(e)\leq \pi(e)$ for each $e\in E$.
 Then there exists
 a feasible solution $(\xi,\nu,\mu)$ to ${\Dual}(I)$.
\end{lemma}
\begin{proof}
 First, we appropriately define $\nu$ and $\mu$.
 All variables of $\nu$ and $\mu$ are initialized to $0$.
 We fix an arbitrary ordering of edges in $E$, and denote the $i$-th edge by $e_i$.

 Let $rv \in E$.
 We sequentially select edges $e_1$ to $e_{|E|}$.
 On selection of $e_i$, we first increase 
 $\mu(r,e_i)$
 until the increase reaches $\xi(e_i)$ or 
 \eqref{eds.d2} becomes tight for $r$.
 If \eqref{eds.d2} is tightened for $r$ before 
 $\mu(r,e_i)$ is increased by $\xi(e_i)$,
 then $\nu(rv,e_i)$ is increased until
 the total increase reaches $\xi(e_i)$ or
 \eqref{eds.d1} becomes tight for $rv$.
Once \eqref{eds.d1} has tightened for $rv$,
 $\mu(v,e_i)$ is increased.
 The current iteration is terminated when the total increase reaches $\xi(e_i)$. If $i < |E|$, the algorithm advances
 to the next iteration, and processes $e_{i+1}$.
 Since $\sum_{i=1}^{|E|}\xi(e_i) \leq \min\{\alpha_1,\alpha_2\} \leq w(rv)+w(r)+w(v)$,
 all edges in $E$ can be processed before \eqref{eds.d2} becomes tight for $v$.

 The above process defines $\nu(rv,e_i)$, $\mu(v,e_i)$, and
 $\mu(r,e_i)$ for each $i \in \{1,\ldots,|E|\}$.
 This process is repeated for all $rv \in E$,
 but $\mu(r,e_i)$ is not increased beyond the first iteration.
 Note that $\mu(r,e_i)$ is assigned the same value 
 regardless of which edge $rv$ we begin with.
 Thus, $\nu$ and $\mu$ have been completely defined, and the feasibility of
 $(\xi,\nu,\mu)$ follows from their definitions.
\end{proof}

\subsection*{Case A}
In this case, a leaf edge $e$ of maximum depth satisfies $\pi(e)>0$.
Since Case A is not the base case, the depth of $l_e$ exceeds one.
Let $u$ denote the upper end node of $e$. Also let $v_0$ denote the parent of
$u$, and let $v_1,\ldots,v_k$ be the children of $u$.
Throughout this paper, the sets $\{1,\ldots,k\}$ and $\{0,\ldots,k\}$ are denoted by $[k]$,
and $[k]^*$, respectively.
The edge joining $v_i$ and $u$ is called $e_i$, with $i\in [k]^*$
($e$ is included in $\{e_1,\ldots,e_k\}$).
The relationships between these nodes and edges are illustrated in Figure~\ref{fig.casea}.
We define $\beta_1=\min_{i=0}^k(w(e_i)+w(u)+w(v_i))$, 
$\beta_2=\sum_{i=1}^k\pi(e_i)$, and $\beta=\min\{\beta_1,\beta_2\}$.
Let $i^* \in [k]^*$ be the index of an edge $e_{i^*}$ that attains $\beta_1=w(e_{i^*})+w(u)+w(v_{i^*})$.

\begin{figure}
 \centering
 \begin{tikzpicture}[line width=1.5pt, label distance = -3pt]

  \node[vertex,label=135:$u$] (u) at (0,0) {};
  \node[vertex,label=left:$v_0$] (v0) at (0,1) {};

  \foreach \i in {1,...,9}
  \node[vertex] (v\i) at (-2 + .4*\i, -1.5) {};

  \node at (-1.6,-1.9) {$v_1$};
  \node at (-1.15,-1.9) {$v_2$};
  \node at (1.7,-1.9) {$v_k$};
  \node at (-.8,-.4) {$e_1$};
  \node at (.8,-.4) {$e_k$};
  \node at (.3,.5) {$e_0$};

  \foreach \i in {0,...,9}
  \draw (u) -- (v\i);

\end{tikzpicture}
 \caption{Edges and nodes in Case~A}
 \label{fig.casea}
\end{figure}
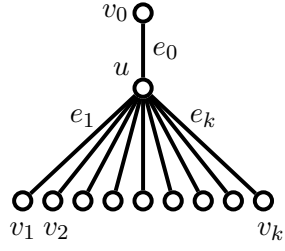

The algorithm constructs an 
instance $I'=(G',w',\pi')$ as follows.
Suppose that $\beta_1> \beta_2$.
In this case, $G'$ is defined as $G$, and
$\pi'\colon E'\rightarrow \Rset_+$ is defined by
\[
\pi'(e)=
\begin{cases}
 0 & \mbox{if } e \in \{e_1,\ldots,e_k\},\\
 \pi(e) & \mbox{otherwise.}
\end{cases} 
\]
For $\psi \in \Rset$, we denote $\max\{0,\psi\}$ by $(\psi)_+$.
The weight function $w'\colon V' \cup E'  \rightarrow \Rset_+$ is
defined by
\[
w'(e)=
 \begin{cases}
  (w(e_i)-(\beta-w(u))_+)_+ & \mbox{if } e=e_i, i \in [k]^*,\\
  w(e) & \mbox{otherwise}
 \end{cases}
\]
for each $e \in E'$, and 
\[
w'(v)=
 \begin{cases}
  (w(u)-\beta)_+ & \mbox{if } v = u,\\
  w(v_i)-(\beta-w(u)-w(e_i))_+ & \mbox{if } v=v_i,  i \in [k]^*,\\
  w(v) & \mbox{otherwise}
 \end{cases}
\]
for each $v \in V'$.
If $\beta_1 \leq \beta_2$, then
$G'=(V',E')$ is defined as the tree obtained by removing nodes
$v_1,\ldots,v_k$ and edges $e_1,\ldots,e_k$ from $G$, and $\pi'\colon E'\rightarrow \Rset_+$ is defined by
\[
\pi'(e)=
\begin{cases}
 0 & \mbox{if } e =e_0,\\
 \pi(e) & \mbox{otherwise.}
\end{cases} 
\]
$w'$ is defined identically to the case $\beta_1 >
\beta_2$, ignoring $w'(v_i)$ and $w'(e_i)$ for $i \in [k]$. 

If $\beta_1\leq \beta_2$, 
the number of nodes with depth exceeding one is lower in $G'$ than in $G$.
Hence, the algorithm inductively finds a solution $F'$ to $I'$ and a feasible dual
solution $(\xi',\nu',\mu')$ to $\Dual(I')$ that satisfy
\eqref{eq.eds-optimiality}.
Otherwise, the number of leaf edges $e$ of maximum depth with
$\pi'(e)>0$ is lower in $G'$ than in $G$. 
If $G'$ lacks edges of this type,
then instance $I'$ is categorized into Case~B,
and $F'$ and
$(\xi',\nu',\mu')$ are found as demonstrated below.
If such edges do exist in $G'$,
the algorithm finds $F'$ and $(\xi',\nu',\mu')$ by induction on the 
number of such edges.
Therefore, it suffices to show that the required $F$ and $(\xi,\nu,\mu)$
can be constructed from $F'$ and $(\xi',\nu',\mu')$, provided that
$F'$ and $(\xi',\nu',\mu')$ exist.

We now define $F$ and $(\xi,\nu,\mu)$.
$F$ is defined by
\[
 F=
 \begin{cases}
  F' \cup \{e_0\}& \mbox{if } \delta_{F'}(v_0)\neq \emptyset, \beta > w(u)+w(e_0),\\
  F' & \mbox{if }  \delta_{F'}(v_0)= \emptyset \mbox{ or }  \beta \leq
  w(u)+w(e_0), \beta_1 > \beta_2, \\
  F'\cup \{e_{i^*}\} & \mbox{if }
 \delta_{F'}(v_0)= \emptyset \mbox{ or } \beta \leq  w(u)+w(e_0), \beta_1 \leq \beta_2. \\
 \end{cases}
\]

\begin{lemma}
There exists a feasible solution $(\xi,\nu,\mu)$ to $\Dual(I)$ that
satisfies \eqref{eq.eds-optimiality} with $F$.
\end{lemma}
\begin{proof}
We first consider the case of $\beta_1 > \beta_2$.
 In this case, 
 $\xi'(e_i)=0$ follows from $\pi'(e_i)=0$ for each $i\in [k]$.
We define 
$\xi(e_i)$ as $\pi(e_i)$ for $i \in [k]$.
We also define $\nu(e_j,e_i)$, $\mu(v_j,e_i)$, and $\mu(u,e_i)$ 
such that $\nu(e_j,e_i)+\mu(v_j,e_i)+\mu(u,e_i)=\xi(e_i)$
holds for each $j  \in [k]^*$ and $i \in [k]$.
 The other dual variables are set to their values assigned in $(\xi',\nu',\mu')$.
Note that, for each
 $i \in [k]^*$, we have
 $w(u)+w(e_i)+w(v_i)-w'(u)-w'(e_i)-w'(v_i)=\beta_2=\sum_{i=1}^k \xi(e_i)$.
Hence, $\nu(e_j,e_i)$, $\mu(v_j,e_i)$, and
 $\mu(u,e_i)$ can be defined without violating \eqref{eds.d1} or \eqref{eds.d2} as
 follows.
 We sequentially collect edges $e_1$ to $e_{k}$.
On selection of $e_i$, 
 $\mu(u,e_i)$ is increased
 until the total increase reaches $\xi(e_i)$ or 
 \eqref{eds.d2} becomes tight for $u$.
 If \eqref{eds.d2} is tightened for $u$ before 
 $\mu(u,e_i)$ has increased by $\xi(e_i)$,
 then $\nu(e_j,e_i)$ is simultaneously increased for all $j\in [k]^*$
 until \eqref{eds.d1} becomes tight for $e_j$.
Once \eqref{eds.d1} has tightened for $e_j$,
$\mu(v_j,e_i)$ is increased instead of $\nu(e_j,e_i)$.

 We next consider the case of $\beta_1 \leq \beta_2$.
 In this scenario, $\xi'(e_0)=0$ holds because $\pi'(e_0)=0$.
 We define $\xi(e_i)$, $i\in [k]$ such that $\xi(e_i)\leq \pi(e_i)$ for
 each $i \in [k]$ and $\sum_{i=1}^k \xi(e_i)=\beta_1$, which
 is possible because $\sum_{i=1}^k \pi(e_i)=\beta_2 \geq \beta_1$.
 $\nu(e_i,e_0)$ and $\mu(v_i,e_0)$ are set to $0$ for each $i\in [k]$.
 We also define $\nu(e_j,e_i)$, $\mu(v_j,e_i)$, and $\mu(u,e_i)$ 
 such that $\nu(e_j,e_i)+\mu(v_j,e_i)+\mu(u,e_i)=\xi(e_i)$ 
 for each $j \in [k]^*$ and $i \in [k]$
 as specified for $\beta_1 > \beta_2$.
 The other variables are set to their values assigned in
 $(\xi',\nu',\mu')$. The feasibility of $(\xi,\nu,\mu)$ follows from its definition.

 We now prove that $F$ and $\xi$ satisfy \eqref{eq.eds-optimiality}.
 Without loss of generality, we can assume $|\delta_{F'}(u)|\leq 1$
 (if this condition is false, we can remove edges $e_i$, where $i\in [k]$, from $F'$ until
 $|\delta_{F'}(u)|=1$).
 The objective value of $F$ exceeds that of $F'$ by at most
 $\beta$,
 unless $e_0 \in F \setminus F'$ and $i^*\neq 0$.
 If $e_0 \in F \setminus F'$ and $i^*\neq 0$, then $\delta_{F'}(v_0)\neq
 \emptyset$ and $\beta > w(u)+w(e_0)$ by the definition of $F$.
 Since $\delta_{F'}(v_0)\neq  \emptyset$, $w'(v_0)$ is counted in 
 the objective value of $F'$.
 Thus, the objective values increases from $F'$ to $F$
by $w(v_0)-w'(v_0) + w(e_0)+w(u)$. From $\beta>w(u)+w(e_0)$, it follows that 
 $w'(u)=w'(e_0)=0$, therefore, the objective function increases by $\beta$.
 Since $\sum_{e \in E}\xi(e) - \sum_{e \in E'}\xi(e)=\beta$,
 \eqref{eq.eds-optimiality} is satisfied.
\end{proof}

\subsection*{Case B}
In this case, $\pi(e)=0$ holds for all leaf edges $e$ of maximum depth.
Let $s$ be the grandparent of a leaf node of maximum depth.
Also, let $u_1,\ldots,u_k$ be the children of $s$, and $e_i$ be the edge joining
$s$ and $u_i$ for $i \in [k]$.
In the following discussion, we assume that $s$ has a parent, and 
that each node $u_i$ has at least one child.
This discussion is easily modified to cases in which 
$s$ has no parent or some node $u_i$ has no child.
We denote the parent of $s$ by $u_0$, and the edge between $u_0$ and $s$ by
$e_0$.
For each $i \in [k]$,
let $V_i$ be the set of children of $u_i$, and
$H_i$ be the set of edges joining $u_i$ to its child nodes in $V_i$. Also define $h_i=u_iv_i$ as an edge that attains 
$\min_{u_iv \in H_i} (w(u_iv)+w(v))$.
The relationships between these nodes and edges are illustrated in Figure~\ref{fig:eds-caseb}.

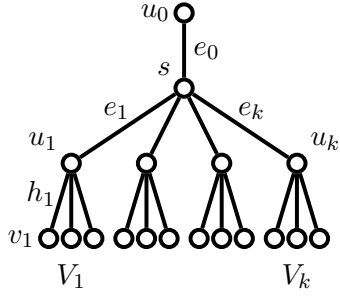
\begin{figure}[t]
 \centering
 \begin{tikzpicture}[line width=1.5pt, label distance = -3pt]

  \node[vertex,label=135:$s$] (s) at (0,0) {};
  \node[vertex,label=left:$u_0$] (u0) at (0,1) {};
  \node[vertex,label=135:$u_1$] (u1) at (-1.5, -1) {};
  \node[vertex,label=left:$v_1$] (v1a) at (-1.8, -2) {};
  \node[vertex] (v1b) at (-1.5, -2) {};
  \node[vertex] (v1c) at (-1.2, -2) {};
  \node[vertex] (u2) at (-.5, -1) {};
  \node[vertex] (v2a) at (-0.8, -2) {};
  \node[vertex] (v2b) at (-.5, -2) {};
  \node[vertex] (v2c) at (-.2, -2) {};
  \node[vertex] (u3) at (.5, -1) {};
  \node[vertex] (v3a) at (0.8, -2) {};
  \node[vertex] (v3b) at (.5, -2) {};
  \node[vertex] (v3c) at (.2, -2) {};
  \node[vertex,label=45:$u_k$] (u4) at (1.5, -1) {};
  \node[vertex] (v4a) at (1.8, -2) {};
  \node[vertex] (v4b) at (1.5, -2) {};
  \node[vertex] (v4c) at (1.2, -2) {};

   \node at (-1.5,-2.5) {$V_1$};
   \node at (1.5,-2.5) {$V_k$};
   \node at (-1.9,-1.4){$h_1$};
   \node at (.9,-.3) {$e_k$};
   \node at (-.9,-.3) {$e_1$};
   \node at (.3,.5) {$e_0$};

  \draw (s) -- (u0);
  \draw (s) -- (u1);
  \draw (s) -- (u2);
  \draw (s) -- (u3);
  \draw (s) -- (u4);
  \draw (u1) -- (v1a);
  \draw (u1) -- (v1b);
  \draw (u1) -- (v1c);
  \draw (u2) -- (v2a);
  \draw (u2) -- (v2b);
  \draw (u2) -- (v2c);
  \draw (u3) -- (v3a);
  \draw (u3) -- (v3b);
  \draw (u3) -- (v3c);
  \draw (u4) -- (v4a);
  \draw (u4) -- (v4b);
  \draw (u4) -- (v4c);
\end{tikzpicture}
 \caption{Edges and nodes in Case~B}\label{fig:eds-caseb}
\end{figure}

Now define $\theta_1=\min_{i=0}^k (w(e_i)+w(u_i)+w(s))$, 
$\theta_2=\sum_{i=1}^k \min\{w(u_i)+w(v_i)+w(h_i), \pi(e_i)\}$,
and let $\theta=\min\{\theta_1,\theta_2\}$.
We denote the index $i \in [k]$ of an edge $e_i$ that attains $\theta_1=w(e_i)+w(u_i)+w(s)$ by $i^*$, 
and specify $K=\{i \in [k] \colon w(u_i)+w(v_i)+w(h_i) \leq \pi(e_i)\}$.

We define $I'=(G',w',\pi')$ as follows.
If $\theta_1\geq \theta_2$,
then $G'$ is the tree obtained by removing all edges in 
$\bigcup_{i \in [k]}H_i$ and all nodes in $\bigcup_{i \in [k]}V_i$ from $G$, 
and $\pi'\colon E' \rightarrow \Rset_+$ is defined such
that
\[
 \pi'(e)=
 \begin{cases}
  0 & \mbox{if } e \in \{e_1,\ldots,e_k\},\\
  \pi(e) & \mbox{otherwise}
 \end{cases}
\]
for $e \in E'$.
In this case, $w'\colon V' \cup E' \rightarrow \Rset_+$ is defined by
\[
 w'(v)=
 \begin{cases}
  (w(s)-\theta)_+ & \mbox{if } v = s,\\
  w(u_i)-(\theta-w(s)-w(e_i))_+ & \mbox{if } v=u_i, i \in [k]^*\\
  w(v) & \mbox{otherwise}
 \end{cases}
\]
for $v \in V'$, and 
\[
 w'(e)=
 \begin{cases}
  (w(e_i)- (\theta-w(s))_+)_+ & \mbox{if } e = e_i,  i \in [k]^*,\\
 w(e) & \mbox{otherwise,}
 \end{cases}
\]
for $e \in E'$.
If $\theta_1 < \theta_2$, then
$e_1,\ldots,e_k$, and their descendants are removed
from $G$ to obtain $G'$, and
$\pi'$ is defined by
\[
 \pi'(e)=
 \begin{cases}
  0 & \mbox{if } e=e_0,\\
  \pi(e) & \mbox{otherwise}.
 \end{cases}
\]
Moreover, $w'$ for $E'$ and $V'$ is defined as in the case $\theta_1
\geq \theta_2$, disregarding the weights of edges and nodes removed from $G'$.

Since $G'$ has fewer nodes of depth exceeding one than $G$, 
the algorithm inductively finds
a solution $F'$ to $I'$, and a feasible solution $(\xi',\nu',\mu')$ to
$\Dual(I')$
satisfying \eqref{eq.eds-optimiality}.
$F$ is constructed from $F'$ as follows.
\[
 F =
 \begin{cases}
  F' \cup \{e_0\} & \mbox{if } \delta_{F'}(u_0)\neq \emptyset, \theta >
  w(s)+w(e_0),\\
  F'& \mbox{if } 
  \delta_{F'}(u_0)= \emptyset \mbox{ or } \theta \leq w(s)+w(e_0),
  \delta_{F'}(s)\neq \emptyset,\\
  F'\cup \{h_i \colon i \in K\} & \mbox{if } 
  \delta_{F'}(u_0)= \emptyset \mbox{ or } \theta \leq w(s)+w(e_0),
  \delta_{F'}(s)= \emptyset,
  \theta_1 \geq \theta_2,\\
  F' \cup \{e_{i^*}\} & \mbox{if } 
  \delta_{F'}(u_0)= \emptyset \mbox{ or } \theta \leq w(s)+w(e_0),
  \delta_{F'}(s)=\emptyset,
  \theta_1 < \theta_2.
 \end{cases}
\]
 We define $\xi(e_1),\ldots,\xi(e_k)$
such that
 $\xi(e_i)\leq \min\{w(u_i)+w(v_i)+w(h_i),\pi(e_i)\}$ for $i \in [k]$
 and $\sum_{i=1}^k \xi(e_i)=\theta$, which is possible because
 $\sum_{i=1}^k\min\{w(u_i)+w(v_i)+w(h_i),\pi(e_i)\} =\theta_2 \geq
 \theta$.
 We also define $\xi(e)=0$ for each $e \in \bigcup_{i=1}^k H_i$.
 The other variables in $\xi$ are set to their values in $\xi'$.
The following lemma states that this $\xi$ can form a feasible solution
 to $\Dual(I)$.

\begin{lemma}\label{lem.eds-caseb}
Suppose that 
 $\xi(e_1),\ldots,\xi(e_k)$ satisfy 
 $\xi(e_i)\leq \min\{w(u_i)+w(v_i)+w(h_i),\pi(e_i)\}$ for each $i \in [k]$
 and $\sum_{i=1}^k \xi(e_i)=\theta$. Further, suppose that
 $\xi(e)=0$ holds for each $e \in \bigcup_{i=1}^k H_i$,
 and the other variables in $\xi$ are set to their values in $\xi'$.
 Then there exist $\nu$ and $\mu$ such that
 $(\xi,\nu,\mu)$ is feasible to $\Dual(I)$.
\end{lemma}

 \begin{proof}
  For $i \in [k]$ and $v \in V_i$,
  we define 
  $\mu(v,e_i)$ and $\nu(u_iv,e_i)$ such that 
  $\mu(v,e_i)+\nu(u_iv,e_i)=\min\{w(v_i)+w(h_i),\xi(e_i)\}$.
  This may be achieved without violating the constraints, 
  because $w(v)+w(u_iv) \geq w(v_i)+w(h_i)$. We also
  define $\nu(u_iv,e_i)$ as $(\xi(e_i)-w(u_i)-w(h_i))_+$.
  These variables satisfy \eqref{eds.d1} for $u_iv$, \eqref{eds.d2} for $v$ and
  $u_i$, and \eqref{eds.d3} for $(e_i,u_iv)$.
  $\nu(e_j,e_i)$ for $i \in [k]$ and $j\in [k]^*$, and $\mu(v,e_i)$ for $i
  \in [k]$ and $v \in \{s\} \cup \{u_j\colon j\in [k]^*,j\neq i\}$ are set to $0$.
  The other variables in
  $\nu$ and $\mu$ are set to their values in $\nu'$ and $\mu'$.
  To advance the proof, we introduce an algorithm that increases
  $\nu(e_j,e_i)$ for $i \in [k]$ and $j\in [k]^*$, and $\mu(v,e_i)$ for $i
  \in [k]$ and $v \in \{s,u_0,\ldots,u_k\}$.
  At the completion of the algorithm, $(\xi,\nu,\mu)$ is a feasible solution to $\Dual(I)$.

 The algorithm performs $k$ iterations, and
 the $i$-th iteration increases the variables
 to satisfy \eqref{eds.d3} for
 each pair of $e_i$ and $e_j$, where $j \in [k]^*$.
 The algorithm retains a set $\Var$ of variables to be increased.
 We introduce a notion of time: Over one unit of time, the
 algorithm simultaneously increases all variables in $\Var$ 
  by one.
 The time consumed by the $i$-th iteration is $\xi(e_i)$.

At the beginning of the $i$-th iteration,
$\Var$ is initialized to $\{\mu(u_j,e_i)\colon j \in [k]^*\}$.
The algorithm updates $\Var$ during the $i$-th iteration as follows.
 \begin{itemize}
  \item At time $(\xi(e_i)-w(v_i)-w(h_i))_+$,
	$\mu(u_i,e_i)$ is added to $\Var$
	if $\Var \neq \{\mu(s,e_i)\}$;
  \item If \eqref{eds.d2} becomes tight for $u_j$ under the increase of $\mu(u_j,e_i) \in \Var$,
	then $\mu(u_j,e_i)$ is replaced by $\nu(e_j,e_i)$ for each $j
	\in [k]^*$;
  \item If $\eqref{eds.d1}$ becomes tight for $e_j$ under the increase of
	$\nu(e_j,e_i) \in \Var$ 
	with some $j \in [k]^*$,
	then $\Var$ is reset to $\{\mu(s,e_i)\}$.
\end{itemize}
 We note that the time spent between two consecutive updates may be
 zero.
  
  $\Var$ always contains a variable that appears in the right-hand side
  of \eqref{eds.d3} for $(e_i,e_j)$ with $j\in [k]^* \setminus \{i\}$,
  and for $(e_i,e_i)$ after time $(\xi(e_i)-w(v_i)-w(h_i))_+$.
  The algorithm updates $\Var$ so that 
  \eqref{eds.d1} and \eqref{eds.d2} hold for all variables except $s$. 
  Hence, to show that $(\xi,\nu,\mu)$ is a feasible solution to $\Dual(I)$,
  it suffices to show that 
  \eqref{eds.d2} for $s$ does not become tight before the algorithm is completed.

  We complete the proof by contradiction. Suppose that 
  \eqref{eds.d2} for $s$ tightens at time $\tau < \xi(e_i)$
  in the $i$-th iteration.
  Since $\Var=\{\mu(s,e_i)\}$ at this moment,
  there exists $j \in [k]^*$ such that
  \eqref{eds.d1} for $e_j$ and \eqref{eds.d2} for $u_j$ are tight.
 The variables in the left-hand sides 
  of \eqref{eds.d1} for $e_j$ and \eqref{eds.d2} for $u_j$ and $s$
  are not simultaneously increased. 
  Nor are these variables increased over time
  $(\xi(e_j)-w(v_j)-w(h_j))_+$ in the $j$-th iteration,
  and $\mu(u_j,e_j)$ is initialized to $(\xi(e_j)-w(v_j)-w(h_j))_+$.
  From this argument, it follows that $w(s)+w(u_j)+w(e_j) <
  \sum_{i'=1}^k \xi(e_{i'}) \leq \theta$.
  However, this result is contradicted by the definition of $\theta$, which implies that $\theta \leq
  \theta_1\leq w(s)+w(u_j)+w(e_j)$. Thus, the claim is proven.
\end{proof}

\begin{lemma}\label{lem.eds-caseb-optimality}
 $F$ and $\xi$ satisfy \eqref{eq.eds-optimiality}.
\end{lemma}
\begin{proof}
For each $i \in [k]$,
either $e_i \not\in E'$ holds, or $\xi'(e_i)=0$ holds (because 
$\pi'(e_i)=0$).
Hence, $\sum_{e \in E}\xi(e)=\sum_{i=1}^k \xi(e_i)+\sum_{e \in
 E'}\xi'(e)=\theta+\sum_{e \in E'}\xi'(e)$.
 Therefore, it suffices to prove that $\sum_{e \in F}w(e) \leq \theta+\sum_{e
 \in F'}w'(e)$.

 Without loss of generality, we can assume $|\delta_{F'}(e_0)|\leq 1$
 (if false, we can remove edges $e_i$, $i\in [k]^*$ from $F'$ until $|\delta_{F'}(e_0)|= 1$).
 In the sequel, we discuss only the case of
 $\delta_{F'}(u_0)\neq \emptyset$ and $\theta > w(s)+w(e_0)$. In the alternative case, the claim immediately follows from the definitions of $F$ and $w'$.
$\delta_{F'}(u_0)\neq  \emptyset$ implies that $w'(u_0)$ is counted in 
 the objective value of $F'$.
 Moreover, $w'(s)=w'(e_0)=0$ follows from $\theta>w(s)+w(e_0)$.
 Thus, the objective values increase from $F'$ to $F$
 by $w(u_0)-w'(u_0) + w(e_0)+w(s)$, which equals $\theta$.
\end{proof}

\begin{proof}[Proof of Theorem~\ref{thm.eds}]
 We have proven that our algorithm always finds a solution $F$ to $I$
 and a feasible solution $(\xi,\nu,\mu)$ to $\Dual(I)$, where both solutions
satisfy \eqref{eq.eds-optimiality}.
By the duality of LPs, the right-hand side of \eqref{eq.eds-optimiality} cannot exceed the optimal value of $I$; thus, $F$ is an optimal solution to $I$.
\end{proof}

\section{Multicut problem in trees}
\label{sec.multicut}

In this section, we prove Theorem~\ref{thm.multicut}.
Again the input tree $G$ is rooted by selecting an
arbitrary root node.
For each $i \in [k]$,
we let $P_i$ denote the path connecting $s_i$ and $t_i$,
and $\lca_i$ denote the maximum-depth common ancestor of $s_i$ and $t_i$.
The paths $P_1,\ldots,P_k$ are called \emph{demand paths}.
We also denote the set of edges in $P_i$ by $E_i$, and the set of nodes
in $P_i$ by $V_i$ for notational convenience.
We say that an edge $e$ \emph{covers} a demand path $P_i$
if $e \in E_i$. The multicut problem in $G$ seeks a minimum
weight set of edges that covers all demand paths.

The prize-collecting multicut problem can be reduced to 
the multicut problem as follows.
For each $i \in [k]$, add new nodes $s'_i,s''_i$ and new edges $s_is'_i,s'_i s''_i$ to $G$,
and replace the $i$-th demand pair by $(s''_i,t_i)$.
Those new nodes and edges are weighted by
$w(s'_is''_i)=\pi(i)$, $w(s_is'_i)=+\infty$,
and $w(s'_i)=w(s''_i)=0$.
Choosing $s'_is''_i$ into a solution to this new instance of the multicut problem corresponds to violating the $i$-th demand in the original instance
of the prize-collecting multicut problem.

Due to this reduction, we consider only
the multicut problem in trees, which is equivalent to assuming that
 $\pi(i)=+\infty$ for all $i \in [k]$.
 For an instance $I=(G,w)$ of the multicut problem, the dual of the LP
 relaxation $\LP(I)$ is given by
  \begin{align}
  &{\TreeDual}(I)=\hspace*{-2em} \notag\\
  &\mbox{maximize} && 
 \sum_{i \in [k]} \xi(i)\notag \\
  &\mbox{subject to} &&
  \xi(i) \leq \nu(e,i)+\mu(u_e,i)+\mu(l_e,i) && \mbox{for } i\in [k], e \in E_i,\label{treedual1}\\
   &&& \sum_{i\in [k]: e \in E_i} \nu(e,i) \leq w(e) && \mbox{for } e \in E,\label{treedual2}\\
   &&& \sum_{i \in [k]: v \in V_i} \mu(v,i) \leq w(v)
   && \mbox{for } v \in V,\label{treedual3}\\
   &&&\xi(i)\geq 0  && \mbox{for } i \in [k],\notag\\
   &&&\nu(e,i) \geq 0 && \mbox{for } i \in [k], e \in E_i,\notag\\
   &&&\mu(v,i) \geq 0&& \mbox{for } i \in [k], v \in V_i.\notag
 \end{align}

Our algorithm initializes the solution set $F$ to an empty set,
and the dual solution $(\xi,\nu,\mu)$ to 0.
The algorithm proceeds in two phases; the increase phase and deletion phase.
The algorithm iterates in the increase phase, selecting edges 
covering demand paths not previously covered by $F$ 
and adding them to $F$
while updating
the dual solution.
The increase phase terminates when all demand paths have been covered by $F$.
In the deletion phase, $F$ is converted into a minimal solution by
removing some edges.

The demand pairs are assumed to be sorted in the decreasing order
of depth of $\lca_i$, implying that $\lca_i$ is not a descendant of $\lca_j$
if $j < i$.

\subsection*{Increase phase}

At the beginning of each iteration in the increase phase, the algorithm
selects the minimum index $i$ for which $P_i$ is not covered by the
current solution $F$. It then 
updates the dual solution $(\xi,\nu,\mu)$, and 
adds several edges to $F$, one of which covers $P_i$.
If all demand paths are covered by $F$ after this operation, the increase phase is terminated and the algorithm proceeds to the deletion
phase; otherwise, it begins the next iteration. 
The iteration that processes $P_i$ is called \emph{the iteration for $i$}.
In the following discussion, we explain 
the update process of the dual solution, and how edges are selected for addition to
$F$ in the iteration for $i$.

First, we define some terminologies.
 We say that $e \in E_i$ is \emph{tight} with regard to $i\in [k]$
 if \eqref{treedual1} becomes an equality for $(i,e)$.
 We say that $e$ is a \emph{bottleneck edge} with regard to $i$
 if it is tight with regard to $i$, if \eqref{treedual2} becomes an
 equality for $e$, and if \eqref{treedual3} becomes an equality for both
 end nodes of $e$.

At the beginning of the iteration,
$\nu$ is assumed to be minimal under the condition that
$(\xi,\nu,\mu)$ is a feasible solution to $\Dual(I)$.
This condition can be assumed without loss of generality because arbitrarily 
decreasing $\nu$ makes it minimal.
By this assumption, if $\nu(e,j)>0$ for some $j \in [i-1]$ and $e \in E_j$, then 
$e$ is tight with regard to $j$.

The algorithm attempts to continuously increase $\xi(i)$.
As in Section~\ref{sec.eds}, we introduce a time interval, during which
$\xi(i)$ increases by one.
To satisfy \eqref{treedual1},
$\nu(e,i)$, $\mu(u_e,i)$ or $\mu(l_e,i)$ 
must be increased at the speed of $\xi(i)$ for each edge $e \in E_i$ that is tight
with regard to $i$.
If no bottleneck edge exists with regard to $i$, 
the algorithm retains an edge set
$H \subseteq \{e \in E_i\colon \sum_{j \in [k]:e \in E_j}\nu(e,j) < w(e)\}$
and a node set
$U \subseteq \{v \in V_i\colon \sum_{j \in [k]:v \in V_j}\mu(v,j) < w(v)\}$
such that each tight edge in $E_i$ is included in $H$
or is incident to a node in $U$.
The algorithm increases
$\nu(e,i)$, $e \in H$, and $\mu(v,i)$, $v \in U$
at the same speed as $\xi(i)$.
We note that $H$ and $U$ are computed
greedily so that they are minimal.

We now explain how the algorithm handles a bottleneck edge $e$.
For an end node $v$ of the bottleneck edge $e$, we define $J(i,v)$ as $\{j \in [i-1] \colon \mu(v,j) > 0\}$. 
The algorithm attempts to decrease $\mu(v,j)$, defined at an end node $v$ of $e$
and $j \in J(i,v)$.
Below we detail how $\mu(v,j)$ is decreased
while retaining the feasibility of $(\xi,\mu,\nu)$.
We note that decrease of $\mu(v,j)$ is not always possible.
We call $v$ \emph{relaxable} (with regard to $i$)
if $\mu(v,j)$ can be decreased for some $j \in J(i,v)$.
If $E_i$ contains bottleneck edges,
the algorithm maintains
\begin{itemize}
 \item a set $R$ of relaxable nodes such that each bottleneck
         edge $e \in E_i$ is incident to at least one node in $R$,
\item  an edge set $H \subseteq \{e \in E_i\colon \sum_{j \in [k]:e \in
         E_j}\nu(e,j) < w(e)\}$,
 \item and a node set
         $U \subseteq \{v \in V_i\colon \sum_{j \in [k]:v \in V_j}\mu(v,j) < w(v)\}$.
\end{itemize}
$R$, $H$, and $U$ are minimal under the condition that
each tight edge is included in $H$, or is incident to a node in $R\cup U$.
The algorithm increases $\xi(i)$,
$\nu(e,i)$ for $e \in H$, 
and $\mu(v,i)$ for $v \in U \cup R$ at the same speed, 
where 
increasing
$\mu(v,i)$ for $v \in R$
involves 
decreasing $\mu(v,j)$ for some $j \in J(i,v)$ and
updating other variables, as explained below.

We now explain how $\mu(v,j)$ is decreased for some $j \in J(i,v)$, and
formally define the relaxability of $v$.
$\mu(v,j)>0$ implies that $E_j$ contains one or two edges incident to
$v$. Suppose that $E_j$ contains a single edge, $f$. 
Let $u$ be the other end node of $f$.
If $f$ is not tight with regard to $j$, then $\mu(v,j)$ is decreased until $f$ becomes
tight.
Even if $f$ is tight with regard to $j$, 
$\nu(f,j)$ or $\mu(u,j)$ 
is increased while $\mu(v,j)$ is decreased
at the same speed, provided that
$f$ is not a bottleneck edge with regard to $j$.
This action retains the feasibility because 
\eqref{treedual2} for $f$ or \eqref{treedual3} for $u$ is not tight
unless $f$ is a bottleneck edge with regard to $j$.
If $f$ is a bottleneck edge with regard to $j$,
the algorithm recursively attempts to decrease $\mu(u,j')$ for some $j' \in
J(j,u)$ if $u$ is relaxable,
and increase $\mu(u,j)$.
Under these update rules, $\mu(v,j)$ is decreased without violating the feasibility of
$(\xi,\mu,\nu)$.
If $E_j$ contains two edges $f$ and $f'$ incident to $v$,
$\mu(v,j)$ decreases only when allowed for both $f$ and $f'$.
We define $v$ as relaxable if one of these updates is possible.
$v$ is not relaxable under the following conditions.

\begin{fact}\label{fact.relaxable}
 A node $v \in V_i$ is not relaxable with regard to $i$ if and only if
 $J(i,v)=\emptyset$, or for each $j \in J(i,v)$, $E_j$ contains a bottleneck edge $f$ 
whose other end node is non-relaxable and which is incident to $v$.
\end{fact}

We note that decreasing $\mu(v,j)$ for
some relaxable node $v \in V_i$ and $j \in J(i,v)$ may cause other variables to increase.
In this case,
if $P_j$ shares nodes or edges with $P_i$,
increasing $\xi(i)$ by $\epsilon >0$ may increase 
the left-hand sides of \eqref{treedual2} and \eqref{treedual3}
by more than $\epsilon$.
Hence, $\epsilon$ must be set sufficiently small that
the feasibility of $(\xi,\mu,\nu)$ is maintained.
In implementing the increase phase, we 
recommend solving an LP for deciding the increment of $\xi(i)$ in
a single step.
The maximum increment 
$\epsilon$ for $\xi(i)$ can be computed by formulating the problem as an LP.

When $\xi(i)$ increases no further, the algorithm adds several edges to
$F$. At this moment, $E_i$ includes a bottleneck edge $e$ 
such that $e$ is tight with regard to all $j \in [k]$, and 
neither of its end nodes are relaxable.
If two or more such edges exist, the edge of maximum depth, denoted $e$, is added to $F$.
We call $e$ the {\em witness} of $P_i$.

The algorithm then completes the following operations for each end node 
$v$ of $e$.
By Fact~\ref{fact.relaxable}, 
$E_j$ contains a bottleneck edge $f$ incident to $v$ for each $j \in
J(i,v)$, where $f=e$ possibly holds.
The algorithm adds such $f$ to $F$ for each $j \in J(i,v)$ with $e \not\in E_j$.
Since the other end node $v'$ of $f$ is non-relaxable,
$E_{j'}$ also contains a bottleneck edge $f'$ incident to $v'$
for each $j' \in J(j,v')$.
If $f$ is added to $F$ and $f \not\in E_{j'}$,
the algorithm adds each of such $f'$ to $F$ and repeats the process
for the other end nodes of $f'$.

\begin{lemma}\label{lem.increasing}
Let $e \in F$. At the completion of the increase phase,
$e$ is a bottleneck edge with regard to each $i \in [k]$.
 Moreover, neither end node of $e$ is relaxable with regard to each $i \in [k]$.
\end{lemma}
\begin{proof}
 When $e$ is added to $F$, it satisfies the above conditions.
 In later iterations, the algorithm does not decrease $\nu(e',i)$ for
 any bottleneck edge $e'$ nor $\mu(v,i)$ for an end node $v$ of $e'$.
 Hence, $e$ satisfies the conditions at completion of the increase phase.
\end{proof}

\subsection*{Deletion phase}
Let $I$ be the set of indices $i$ for which $P_i$ was considered in the
increase phase. In other words, 
$i$ was the minimum index for which $P_i$ was covered by no edge in $F$
at the beginning of some iteration of the increase phase.
The deletion phase is also iterated, sequentially processing $P_{i}$ in
the decreasing order of $i \in I$. 
As defined in the increase phase,
the iteration that considers $P_i$ is called \emph{the iteration for $i$}.
Briefly, the deletion phase selects
edges from $F$ to obtain a final solution $F'$ 
in which each $\xi(i)$, $i \in I$ contributes to at most two edges.

$F'$ is initialized to the empty set. Suppose that the algorithm iterates $i \in I$. 
We denote $\{v \in V_i \colon \mu(v,i)>0, \delta_{F'}(v)\neq
\emptyset\}$ by $\tilde{V}_i$.
Let $v$ be a node in $\tilde{V}_i$ with no ancestor in $\tilde{V}_i$.
By Lemma~\ref{lem.increasing},
 $\delta_{F}(v)\supseteq \delta_{F'}(v) \neq \emptyset$
implies that $v$ is non-relaxable with regard to $i$.
Suppose now that an edge in $\delta_{F'}(v)$ is added to $F'$ in the
iteration for $i'$ with $i' > i$.
$\mu(v,i) > 0$ implies that
$i \in J(i',v)$, and hence by Fact~\ref{fact.relaxable},
$E_i$ contains a bottleneck edge
$f$ incident to $v$, whose other end node is 
non-relaxable.
If $E_i$ contains two such edges, and if $v$ is the lower end node of
one of those edges, 
$f$ becomes the edge satisfying $l_f=v$.
Otherwise, $f$ is arbitrarily selected from the candidate edges.
If $f$ has not previously belonged to $F'$, it is added to $F'$.
If $F'$ contains an edge $f' \in E_i$ which is a descendant of $f$, 
$f'$ is deleted from $F'$. 

 The above $v$ can be selected from at most two choices in $\tilde{V}_i$.
The algorithm completes the above operation for each of these nodes.
At the end of the operation, if
$F'\cap E_i$ does not contain
the witness $e_i$ of $P_i$ or
any ancestor of $e_i$,
the algorithm adds $e_i$ to $F'$.
The algorithm performs no further tasks in the iteration for $i$.
Note that $1 \in I$ always holds.
If $i > 1$, the algorithm proceeds to the next iteration. 
If $i=1$, the deletion phase terminates, and the algorithm outputs $F'$.

\begin{lemma}\label{lem.multicut-key}
Suppose that the algorithm outputs an edge set $F'$. Then it satisfies the following conditions.
\begin{itemize}
 \item[\rm (i)]  $F' \cap E_j\neq \emptyset$ for each $j \in [k]$.
 \item[\rm (ii)] Let $i \in I$. Each subpath between an end node of
	        $P_i$ and $\lca_i$ includes at most one edge in $F'$.
\item[\rm (iii)] Let $i,i' \in I$ with $i \neq i'$.
	       If an edge in $F' \cap E_{i'}$ is incident to a
	       node $v$ in $V_{i}$, and if $\mu(v,i)>0$,
	       then $F' \cap E_{i}$ contains an edge incident to $v$.
\end{itemize}
\end{lemma}
\begin{proof}
 To prove (i),
 we first consider the case of $j \in I$.
 The deletion phase adds
 an edge $e  \in E_j$ to $F'$ before completing the iteration for $j$.
 If $e$ is removed from $F'$ in a later iteration for $j'$,
 then $e \in E_{j'}$ and
 the algorithm adds an ancestor $f \in E_{j'}$ of $e$ to $F'$.
 $\lca_{j'}$ is a descendant of $\lca_j$
 because $j' < j$ and $e \in E_{j'}\cap E_{j}$.
 Hence, $f$ also covers $P_j$. 
 Even if $f$ is eventually removed from $F'$,
 we can similarly prove that $P_j$ is covered by 
 another edge that is added to $F'$ in the iteration.
 Hence, at completion of the algorithm, $P_j$ is covered by $F'$ 
 if $j \in I$.

 Before discussing the case of $j \not\in I$,
 we note that if $j \in I$,
 then $F'\cap E_j$ contains 
 the witness $e_j$ of $P_j$ or one of its ancestors
 at completion of the algorithm.
 This is guaranteed by the deletion phase during iteration for $j$, which adds $e_j$ to $F'$
 if $F'$ does not contain $e_j$ or any of its ancestors.

 We now discuss the case of $j \not\in I$.
 By definition of the increase phase,
  $F$ contains an edge $e$ that covers $P_j$, and $e$ is the
 witness of $P_{i}$ for some $i \in I$
 with $i < j$.
 By the above observation, $F'$ contains $e$ or
 one of its ancestors $f \in E_i$.
 $f$ also covers $P_j$, because $i < j$.
 Hence, $F'$ contains an edge covering $P_j$ even if $j \not\in I$.

 Next, we prove (ii) by induction on $i$.
 Let us consider the case of maximum $i$ in $I$.
 The first iteration of the deletion phase adds the witness $e_i$ of $P_i$ to $F'$.
 Suppose that another edge $f \in E_i$ 
 is added to $F'$ 
  in the iteration for $i' \in I$ in the deletion phase.
 If $f$ is added to $F$ before the iteration for $i$ in the
 increase phase, 
 then $i$ does not belong to $I$, which is a contradiction. 
 Otherwise,
 $f$ is added to $F$ because it is incident to an end node
 of $e$, and $e$ does not cover $E_{i'}$.
 Since $i' <i$, this implies that 
 either $f$ is a descendant of $e$ or $u_e=u_f$.
 In the former case,
 $e$ is not chosen as the witness of $P_i$, because
 the witness is a bottleneck edge of maximum depth with both
 end nodes non-relaxable, which contradicts the definition of $e$.
 In the latter case, $u_e=u_f=\lca_i$, consistent with (ii) in Lemma 6.

 We next consider the case of non-maximal $i$ in $I$.
 Let $K$ be the set of edges in the subpath between an end node of $P_i$ and $\lca_i$.
 Suppose that at the start of the iteration for $i$ in the deletion phase, there exist distinct edges $e, e'\in K\cap F'$.
 Without loss of generality, we assume that $e'$ is an ancestor of $e$.
 Assume that the deletion phase adds
 $e$ to $F'$ during the iteration for $j$ with $i < j$,
 and adds $e'$ to $F'$ during the iteration for $j'$ with $i <j'$.
 Since $\lca_{i}$ is a descendant of $\lca_{j}$,
 $e'$ covers $P_i$; therefore,
 $e'$ also covers $P_j$.
 Hence, the subpath of $P_j$ between $l_e$ and $\lca_j$
 is covered by two edges in $F'$, which is a contradiction by 
 induction.
 Therefore, at the beginning of the iteration
 for $i$ in the deletion phase, $K$ is covered by at most one edge in $F'$.
 
At the start of the iteration for $i$, suppose that $K$ is covered by an edge $e \in F'$.
 If another edge $f \in K$ is added to $F'$ during this iteration,
 then $f$ is incident to a node $v \in \tilde{V}_i$,
and another edge $f'$ incident to $v$ is in $F'$.
 Suppose that $f'$ is added to $F'$ during the iteration for $j$ in the˚
 deletion phase.
 If $v$ is an ancestor of $e$,
 then $e$ is removed from $F'$ when $f$ is added to $F'$.
 If $v$ is a descendant of $e$, then
 $P_j$ is covered by $e$ because $j > i$. 
 Therefore, the subpath of $P_j$ between $l_{f'}$ and
 $\lca_j$
 is covered by both $e$ and $f'$,
 which again is a contradiction by induction.
 
 $K$ could also be covered by two edges if the witness $e_i$ of $P_i$ is in $K$ and is added to $F'$ in
 the iteration for $i$, even though one of its descendants has already been in $F'$.
 We now demonstrate that this situation does not occur. Note that each edge in $F \cap E_i$ is 
 either $e_i$ or is incident to a node in $V_{i'}$ for
 some $i' \in I$ with $i' > i$.
 Since $P_{i'}$ is not covered by $e_i$ and
 $\lca_{i}$ is not an ancestor of $\lca_{i'}$,
 $e_i$ is not an ancestor of
 any edge in $F \cap E_i$.
 Hence, no descendant of $e_i$ is added to $F'$.
 
 Finally, we prove (iii). 
 Suppose that some $i,i' \in I$ violates the claim of (iii).
 We consider such a pair of $i$ and $i'$ that minimizes $|i-i'|$.
 Let $e \in F' \cap E_{i'}$ be an edge incident to a node $v \in
 V_{i}$. Suppose that $i' > i$.
 If $F' \cap E_{i}$ contains no edge incident to $v$,
 then the iteration for $i$ during the deletion phase adds 
 an ancestor edge of $v$ in $E_{i}$ to $F'$. 
 This edge also covers $P_{i'}$ because $\lca_{i'}$ is
 an ancestor of $\lca_{i}$, which contradicts (ii). 
 Next, we suppose that $i' < i$.
 Claim (iii) is obvious when $e$ covers $P_{i}$; thus, we
 suppose that $e$ does not cover $P_{i}$.
 $\mu(v,i)>0$ indicates that $e$ was not in $F$
 at the start of  the iteration for $i$ in the increase phase.
 Let $j \in I$ be the index for which
 $e$ is added to $F$ during the iteration for $j$ in the increase phase,
where $j \geq i$.
 By the definition of the increase phase,
 the witness of $P_j$ is incident to $v$,
 and because $\mu(v,i)>0$, an edge incident to $v$ in $E_i$ is
added to $F$ simultaneously with $e$.
 Since this edge is not in $F'$,
 the iteration for $i$ in the increase phase
 adds an ancestor edge of $v$ in $E_i$ to $F'$.
 This edge covers $P_j$, indicating that the subpath of $P_j$ between its
 one end node and $\lca_j$ is covered by two edges if $j > i$, which contradicts (ii). Therefore, $j=i$.
 However, $e$ is not added to $F'$ unless the witness of $P_i$ is added
 to $F'$ and a contradiction arises.
\end{proof}

\begin{proof}[Proof of Theorem~\ref{thm.multicut}]
 Let $F'$ denote the edge set output by the algorithm.
 By claim (i) of Lemma~\ref{lem.multicut-key}, $F'$ is a multicut.
 Since $\sum_{i \in [k]}\xi(i)$ is a lower bound on the optimal value,
 it suffices to show that $\sum_{e\in F'}w(e) + \sum_{v \in V(F')}w(v) \leq 2 \sum_{i \in [k]}\xi(i)$.

 Lemma~\ref{lem.increasing} implies that 
 $w(e)=\sum_{i \in [k]:e \in E_i}\nu(e,i)$ 
for each $e \in F$,
 and that
 $w(v)=\sum_{i \in [k]:v \in V_i}\mu(v,i)$
 for each $v \in V(F)$.
 Recall that $F' \subseteq F$.
 Hence, 
\begin{equation}\label{eq.cost}
\sum_{e \in F'}w(e)+\sum_{v \in V(F')}w(v) 
 = \sum_{i \in [k]}
 \left(
 \sum_{e \in F' \cap E_i}\nu(e,i)
 +
 \sum_{v \in V(F') \cap V_i}\mu(v,i)
 \right).
\end{equation}
 (iii) of Lemma~\ref{lem.multicut-key} indicates that
 $V(F')\cap V_i \subseteq V(F' \cap E_i)$ holds if $i \in I$.
 If $i \not\in I$, then $\nu(e,i)=0$ for any $e \in E_i$
 and $\mu(v,i)=0$ for any $v \in V_i$.
 Hence, the right-hand side of \eqref{eq.cost} is equal to
 \[
 \sum_{i \in I}
\sum_{e=uv \in F' \cap E_i}
 \left(
 \nu(e,i)+\mu(v,i)+\mu(u,i)
 \right)
 =
 \sum_{i \in I}
 |F' \cap E_i|\xi(i).
 \]
 $|F'\cap E_i| \leq 2$ for each $i \in I$ by (ii) of
 Lemma~\ref{lem.multicut-key}.
 Recall that $\xi(i)=0$ for each $i \not\in I$.
 Therefore, the right-hand-side of (9) is at most $2\sum_{i \in [k]}\xi(i)$.
 \end{proof}

\section{Conclusion}
\label{sec.conclusion}

In this paper, we emphasized a large integrality gap when the natural LP relaxation is applied to the graph
covering problem that minimizes node weights.
We then formulated an alternative LP relaxation for graph covering problems
in edge- and node-weighted graphs that is stronger than the natural
relaxation.
This relaxation was incorporated into an exact
algorithm for the prize-collecting EDS problem in trees, and a 2-approximation algorithm
for the multicut problem in trees. The approximation guarantees for these algorithms match the previously known best
results for purely edge-weighted graphs.
In many other graph covering problems, the integrality gap in the proposed relaxation 
would increase if node weights were introduced, because 
the problems in node-weighted graphs admit stronger hardness results, as shown 
in the Appendix.
Nonetheless, the proposed relaxation is a potentially useful tool for
designing heuristics or using IP solvers to solve the above problems.

\section*{Acknowledgements}

This work was partially supported by Japan Society for the Promotion of
Science (JSPS), Grants-in-Aid for Young Scientists (B) 25730008.
The author thanks an anonymous referee of SWAT 2014
for pointing out an error on the reduction from
the prize-collecting multicut problem to the multicut problem in an earlier version of this paper.

\bibliographystyle{mystyle}
\bibliography{sndp,eds} 

\appendix
\section{EDS problem in general graphs}\label{sec.reduction}
\subsection{Hardness results}
In this section, we discuss the EDS problem
in general graphs.
First, we show that the EDS problem in node-weighted graphs is as hard as 
the set cover problem, and the EDS problem in edge- and node-weighted graphs
is as hard as the (non-metric) facility location problem.
To this end, 
we reduce the set cover problem or the facility
location problem 
to the EDS problem in bipartite graphs.
Accordingly, our hardness results hold for the EDS problem in bipartite graphs.

We now define the set cover problem and facility location problem.
In the set cover problem, we are given a set $V$, family $\mathcal{S}$
of subsets of $V$, and cost function $c\colon \mathcal{S} \rightarrow
\Rset_+$.
A solution is a subfamily $\mathcal{X}$ of $\mathcal{S}$ such that
$\cup_{X \in \mathcal{X}}X=V$. The objective is to minimize $\sum_{X \in \mathcal{X}}c(X)$.
Inputs in the facility location problem
are a client set $V$, a facility set $F$, opening costs
$o\colon F \rightarrow \Rset_+$, and connection costs $d\colon V \times
F \rightarrow \Rset_+$.
A solution to this problem
is a pair of  $F' \subseteq F$ and $\rho\colon V \rightarrow F'$,
and the objective is to minimize
$\sum_{f \in F'}o(f) + \sum_{v \in V}d(v, \rho(v))$.

Given an instance of the set cover problem,
define a facility set $F$ whose members each corresponds to 
a set in $\mathcal{S}$;
We let $X_f$ denote the member of $\mathcal{S}$ corresponding to $f \in F$.
Define opening costs $o$ and connection costs $d$ such that 
$o(f)=c(X_f)$, $d(v,f)=0$ if $v \in X_f$, and
$d(v,f)=+\infty$ otherwise.
Then, the instance $(V,F,o,d)$ of the facility location problem is
equivalent to the instance $(V,\mathcal{S},c)$ of the set cover problem, implying
that the non-metric facility location problem generalizes the
set cover problem.

\begin{theorem}\label{thm:hardness}
If the EDS problem in node-weighted bipartite graphs
admits a $\gamma$-approximation algorithm,
 then the set cover problem also admits a $\gamma$-approximation algorithm.
 If the EDS problem in edge- and node-weighted bipartite graphs
 admits a $\gamma'$-approximation algorithm,
 then the facility location problem also admits a
 $\gamma'$-approximation algorithm.
\end{theorem}
\begin{proof}
To prove Theorem 4, we reduce the facility location problem to the EDS problem.
From a client set $V$ and facility set $F$,
we construct the complete bipartite graph with
bipartition $V$ and $F$.
Moreover, for each $v \in V$,
we add a new node $v'$ and an edge $e_v$ that joins $v$ and $v'$.
Note that this operation retains the bipartite property of the graph.
We define edge weights $w$ such that $w(e)=d(v,f)$ for each $e=vf$, where $v \in V$ and $f \in F$ and $w(e_v)=0$ for each $v \in V$. 
The node weights $w$ are defined by
$w(v)=0$ and $w(v')=+\infty$ for each $v \in V$,
and $w(f)=c(f)$ for each $f \in F$.

Let $S$ be an EDS for this graph.
We can assume $e_v \not\in S$ because $w(v')=+\infty$.
Since $S$ dominates each $e_v$,
$S$ contains an edge incident to each $v \in V$.
Moreover, if $S$ contains more than one edge incident to $v \in V$,
$S$ remains an EDS if one of these edges is arbitrarily discarded.
Hence $S$ contains exactly one edge incident to $v \in V$,
which joins $v$ and a node $f \in F$.
Let $\rho(v)$ be the opposite end node of the edge in $S$ incident to
$v$, and let $F'=\{\rho(v)\colon v \in V\}$.
Then, $(F',\rho)$ is a solution to the facility location
problem, with cost equaling the weight of $S$.
Conversely, given a solution to the facility location problem,
 define $S$ as $\{v\rho(v) \colon v \in V\}$. Then $S$ is an EDS of the
 graph, and the weight of $S$ equals the cost of the solution to
 the facility location problem. Hence, the reduction is that required 
 in the latter part of the theorem.

 As observed above, the set cover problem corresponds to 
 the instances of the facility location problem with zero connection cost.
 For these instances, the reduction defines instances of the EDS
 problem in which all edges are weighted zero.
 The former part of the theorem follows from this statement.
\end{proof}

Without describing the details, 
we can also show that 
reductions in the proof of Theorem~\ref{thm:hardness} are applicable (with modification) to fundamental covering problems 
such as the Steiner tree problem, $T$-join problem, and edge
cover problem.

\subsection{$O(\log |V|)$-approximation algorithm}

Since it is NP-hard to achieve the $o(\log |V|)$-approximation in the set
cover problem, by Theorem~\ref{thm:hardness}, the same situation exists
for the EDS problem in node-weighted bipartite graphs.
Here, we propose an $O(\log |V|)$-approximation algorithm for the
prize-collecting EDS
problem in general edge- and node-weighted graphs.

Our algorithm reduces the prize-collecting EDS problem to the edge cover problem.
Given an undirected graph $G=(V,E)$ and a set $U$ of nodes, 
we define an edge cover as a set $F$ of edges such that $\delta_F(v)\neq
\emptyset$ for each $v \in U$.
Given $w\colon E \cup V \rightarrow \Rset_+$, the problem seeks an
edge cover $F$ that minimizes $w(F) + w(V(F))$.
For an instance $I'=(G,U,w)$ of the edge cover problem, we apply the following LP relaxation:
\begin{align*}
 & \EC(I')= && \\
 &\mbox{minimize} && 
 \sum_{e \in E} w(e)x(e) + \sum_{v \in V}w(v)x(v)\\
 &\mbox{subject to} &&
 \sum_{e \in \delta(v)} x(e) \geq 1 && \mbox{for } v \in U,\\
 &&& x(v) \geq x(e) && \mbox{for } v \in V, e \in \delta(v),\\
 &&& x(e) \geq 0 && \mbox{for } e \in E,\\
 &&& x(v) \geq 0 && \mbox{for } v \in V.
\end{align*}

\begin{lemma}\label{lem.edstoec}
Suppose that an algorithm computes an edge cover
of weight at most $\gamma \EC(I')$ for any instance $I'$ of the edge
 cover problem.
Then the prize-collecting EDS problem in edge- and node-weighted graphs admits
 a $4\gamma$-approximation algorithm.
\end{lemma}
\begin{proof}
We first solve $\LP(I)$ for the given instance $I=(G,w,\pi)$ of the
prize-collecting EDS problem.
Let $(x,y,z)$ be the obtained optimal solution to $\LP(I)$.
Recall that $\Efam=\{\delta(e)\colon e \in E\}$ in the EDS problem; thus we denote $y(C,e)$ and $z(C)$ by $y(e',e)$ and $z(e')$ 
respectively, where $e'$ is the edge corresponding to $C \in \Efam$.
Let $U=\{v \in V \colon \sum_{e \in \delta(v)}x(e)\geq 1/4\}$.
An instance $I'$ of the edge cover problem is assumed to
consist of $G$, $U$, and $w$.

We now prove that $\EC(I') \leq 4(\sum_{e \in E}w(e)x(e)+\sum_{v \in V}w(v)x(v))$.
In the following, $4x$ is denoted by $x'$. 
Since $\sum_{e \in \delta(v)}x(e)\geq 1/4$ for each $v \in U$,
$\sum_{e \in \delta(v)}x'(e)\geq 1$ for each $v \in U$, implying that $x'$
satisfies the first constraints of $\EC(I')$.
By the minimality of $x$,
there exists $e' \in \delta(e)$ such that $x(e)=y(e',e)$ for any $e \in E$.
The second constraint for $v$ and $e'$ in $\LP(I)$
implies that $x'(v)=4x(v)\geq 4y(e',e)=4x(e)=x'(e)$ for $e \in E$ and an end node
 $v$ of $e$.
Hence, $x'$ satisfies the second constraints of $\EC(I')$.
Therefore, $x'$ is a feasible solution to $\EC(I')$, and 
$\EC(I')$ is at most $4(\sum_{e \in E}w(e)x(e)+\sum_{v \in V}w(v)x(v))$.

Let $F$ be an edge cover of weight at most $\gamma \EC(I')$ for $I'$.
Let $D=\{e \in E \colon z(e)\leq 1/2\}$.
By the first constraint of $\LP(I)$,
$\sum_{e' \in \delta(e)}y(e,e')\geq 1/2$ for each $e \in D$.
Let $u$ and $v$ be the end nodes of $e \in D$.
Since 
$\sum_{e' \in \delta(e)}y(e,e')
 \leq \sum_{e' \in \delta(u)}y(e,e')+\sum_{e' \in \delta(v)}y(e,e')
 \leq \sum_{e' \in \delta(u)}x(e')+\sum_{e' \in \delta(v)}x(e')$, 
at least one of $u$ and $v$ is included in $U$.
Note that each node in $U$ has some incident edge in $F$.
Hence, each edge in $D$ is dominated by $F$.
Therefore, $F$ must pay at most 
$\sum_{e \not\in D}\pi(e) \leq 2\sum_{e \in E}\pi(e)z(e)$
as the penalty.
Summing up, the objective value of $F$ for $I$ does not exceed
\[
\gamma \EC(I')  + 2\sum_{e \in E}\pi(e)z(e)
 \leq 
 4\gamma \left(\sum_{e \in E}w(e)x(e)+\sum_{v \in V}w(v)x(v)\right)
 + 2\sum_{e \in E}\pi(e)z(e)
\leq 4\gamma \LP(I).
\]
Since $\LP(I)$ is a lower bound on the optimal value of $I$,
$F$ achieves a $4\gamma$-approximation for $I$.
\end{proof}

To obtain an $O(\log |V|)$-approximation algorithm for the
prize-collecting EDS problem,
it suffices to obtain an algorithm for the
edge cover problem required in Lemma~\ref{lem.edstoec} with
$\gamma=O(\log |V|)$.
As a sequel, we observe that such an algorithm indeed exists.

Since each node in $U$ is included in $V(H)$ for any edge cover $H$,
an algorithm that solves instances of $w(v)=0$ for all
$v \in U$ is sufficient.
Moreover, 
if $G$ contains an edge $e$ joining nodes $u$ and $v$ in $U$,
the instance is transformed into an equivalent instance by
inserting a new node $s$ that subdivides the edge $uv$ and 
setting the new node and edge weights as $w(us)=w(sv)=0$ and $w(s)=w(e)$.
In the obtained instance, 
$U$ forms an independent set.
Such instances of the edge cover problem are included in the (non-metric)
facility location problem. In fact, we may regard $U$ and $V\setminus U$ as the client 
and facility sets, respectively.
The weights of edges between clients and facilities represent 
connection costs,
and the weights of clients indicate opening costs.
Each edge in edge covers joins a client and facility and 
naturally allocates the client to the facility.

Consider an instance of the facility location problem.
We define a \emph{star} as a set comprising a facility $f \in F$, 
and clients $v_1,\ldots,v_k \in V$.
The cost of the star is $o(f)+\sum_{i=1}^k d(v_i,f)$.
Let $\mathcal{S}$ be the set of all stars.
Identifying the star as the subset $\{v_1,\ldots,v_k\}$ of $V$,
the facility location problem can be regarded as the set cover
problem with the set $V$ and family $\mathcal{S}$ of subsets of $V$.
Hence, the greedy algorithm for the set cover problem achieves $O(\log
|V|)$-approximation for the facility location problem.
An analysis of this greedy algorithm \cite{Lovasz1975,Chvatal1979}
has shown that the costs of solutions are bounded
by an LP relaxation of the set cover problem
(see also \cite{Vazirani2001}).
Our LP relaxation $\EC(I')$ applied to an instance $I'$ of the edge cover
problem is equivalent to this LP relaxation for the
set cover problem derived from $I'$.
The runtime of the greedy algorithm is a primary concern, because
the size of $\mathcal{S}$ in the set cover problem
is not bounded by a polynomial of the input size of $I'$.
The greedy algorithm can operate
in polynomial time for the given instances in the facility location
problem and edge cover
problem~\cite{Hochbaum1982}. Therefore, we can state the following theorem.

\begin{theorem}
 The prize-collecting EDS problem in general edge- and node-weighted graphs admits an $O(\log
 |V|)$-approximation algorithm.
\end{theorem}
\end{document}